\documentclass[aps,prx,reprint,superscriptaddress,amsmath,amssymb]{revtex4-2}

\usepackage{amsmath,amssymb,amsthm,mathtools}
\usepackage{bm}

\usepackage{graphicx}

\newtheorem{theorem}{Theorem}
\newtheorem{lemma}{Lemma}

\newcommand{\cX}{\mathcal{X}}

\newcommand{\KL}{D_{\mathrm{KL}}}
\newcommand{\Var}{\mathrm{Var}}
\newcommand{\E}{\mathbb{E}}
\newcommand{\Prb}{\mathbb{P}}

\newcommand{\hb}{h} 

\begin{document}

\title{{Slowdown and saturation of internal time according to the statistics of information input: a minimal model of response systems}}

\author{Tatsuaki Tsuruyama}
\email{tsuruyam@kuhp.kyoto-u.ac.jp}
\affiliation{Department of Physics, Tohoku University, Sendai 980-8578, Japan}
\affiliation{Department of Drug Discovery Medicine, Kyoto University, Kyoto 606-8501, Japan}

\date{\today}

\begin{abstract}
We consider a response system that updates its internal state in accordance with information input arriving from outside. In this paper, we define as internal time the ``number of kinds'' of codes that have been observed at least once up to a given time, and analyze how the way internal time advances is determined by the statistics of information input (arrival rate and code distribution). When arrivals follow a Poisson process, the average advancing speed of internal time decreases monotonically with time, and if the number of kinds of codes is finite, it eventually approaches an upper limit and saturates. As a result, on long time scales, internal time becomes relatively shorter than physical time.
For a uniform code distribution, we provide a closed form for the correspondence between internal time and physical time, and show that the physical time required to ``advance internal time by one step'' increases in later stages. As an ancillary quantity, we quantify by conditional entropy the remaining uncertainty of ``which codes have been observed'' when only internal time is known, and we give unimodality and the maximization time in the uniform case, and upper bounds, equality conditions, and expressions of the difference from the upper bound in the non-uniform case. Finally, we also present a generalization that assigns weights (description lengths) to each code so that internal time is ticked according to the amount of information in the input.

\end{abstract}

\maketitle
\section{Introduction}
Even if the same physical time elapses, it is not necessarily the case that the updating of experience or the progress of learning is constant. In a situation where, among information inputs arriving from outside, the codes that have not yet been observed become fewer, the opportunities for updating themselves decrease; updating slows down and eventually heads toward saturation. This paper formulates this phenomenon as a minimal model of a ``response system'' that updates its internal state according to the statistics of information input, and rigorously derives that \emph{internal time slows down relative to physical time, and saturates if the number of kinds of codes is finite}.

In psychology and cognitive science, there are various studies on time evaluation \cite{Friedman2010,BlockZakayHancock1998}, but this paper does not aim to directly explain those measurement indices. Instead, under a simple setting in which information input arrives from finitely many codes and the response system keeps only a minimal record, we show that a structure in which the advancing speed of internal time decreases with time necessarily emerges from input statistics (arrival rate, code distribution, and weighting). Thus the central conclusion of this paper is that ``the internal time ticked by a response system slows down according to input statistics, and saturates under a finite number of codes.'' That is, internal time is not a mere mapping of physical time, but a response quantity determined by input statistics, and on long time scales internal time becomes relatively shorter than physical time.

\paragraph{Minimality of the response system}
A response system in this paper is a system that updates its internal state whenever an information input arrives and keeps a minimal record for subsequent response. In this paper, the retained record is restricted to ``the kinds of codes observed so far,'' and the number of such kinds is called internal time. Due to this minimality, the slowdown and saturation of internal time appear in a form that can be computed immediately from input statistics.

\paragraph{Organization}
In \S\ref{sec:model} we define the model, and in \S\ref{sec:slowdown} we present slowdown and saturation of internal time as the main result.
In \S\ref{sec:weighted} we describe a weighted generalization of internal time.
In \S\ref{sec:subjective_time} we organize the correspondence between physical time and internal time, and clarify that internal time becomes harder to advance in later stages.
In \S\ref{sec:conditional_entropy} we introduce conditional entropy as an ancillary quantity, and present the closed form, unimodality, and the maximization time in the uniform case, and upper bounds, equality conditions, KL expressions, and approximate evaluations in the non-uniform case.
Finally, in \S\ref{sec:discussion} we discuss the meaning, limitations, and directions of extension.

\section{Model}
\label{sec:model}

In this paper, we call the kinds of external input \textbf{codes}. We call the variable that summarizes, without duplication, the codes that have been observed at least once up to a given time the \textbf{list of observed codes}, and denote it by $S(t)$. We call the number of codes contained in the list of observed codes \textbf{internal time}, and denote it by $\tau(t)$. Arrival times are described by an \textbf{arrival process}, and a code is observed at each arrival. The input statistics referred to in this paper are (i) the arrival rate $\lambda$ (the mean number of arrivals per unit time) and (ii) the code distribution $p=(p_1,\dots,p_K)$ (the relative frequency with which each code is observed). As a further generalization, we introduce (iii) code-dependent weights $\ell_i$ (description lengths).

\subsection{Information input: arrival process and codes}
Information input consists of an arrival-time sequence $\{t_n\}_{n\ge 1}$ and a code sequence $\{X_n\}_{n\ge 1}$.
The possible values of the codes are finite:
\[
\mathcal{X}=\{1,2,\dots,K\}.
\]
Let the number of arrivals up to time $t$ be
\[
N(t):=\#\{n:t_n\le t\}.
\]
As a standard setting, we assume that $N(t)$ is a Poisson process with intensity $\lambda>0$, and that $X_n$ are i.i.d.\ with distribution $p=(p_1,\dots,p_K)$.
For general theory of point processes, see \cite{DaleyVereJones2003,Kingman1993}.

\subsection{Internal state of the response system: list of observed codes and internal time}
Define the variable that summarizes, without duplication, the codes observed at least once up to time $t$ by
\[
S(t):=\{X_1,\dots,X_{N(t)}\}\subseteq \cX
\]
(here we use the mathematical notation \(\{\cdot\}\)).
This is equivalent to the bit sequence indicating whether each code $i$ has been observed:
\begin{equation}
I(t)=(I_1(t),\dots,I_K(t))
\end{equation}
\begin{equation}
I_i(t):=\mathbf{1}\{i\in S(t)\}\in\{0,1\}.
\end{equation}

Define internal time by
\[
\tau(t):=|S(t)|=\sum_{i=1}^K I_i(t).
\]
$\tau(t)$ is an integer-valued stochastic process that represents the ``number of kinds of observed codes.''

\paragraph{Meaning of internal time}
Internal time $\tau(t)$ is not a quantity that counts ``how many times inputs arrived'' like the arrival count $N(t)$, but a quantity that counts ``how many kinds of codes were observed.'' Even with the same number of arrivals, if the same code is repeatedly observed, $\tau(t)$ hardly increases.
This difference is the direct cause of the slowdown and saturation shown later.

\section{Main result: slowdown and saturation of internal time}
\label{sec:slowdown}

In this section we show that, under Poisson arrivals, internal time $\tau(t)$ slows down and, for finite $K$, saturates.
The conclusions here are rigorous results on the mean behavior of internal time in the response system.
The later section on the ancillary quantity is introduced independently of this main result.

\subsection{General distribution  (Poisson arrivals)}
Assume Poisson arrivals (rate $\lambda$) and an i.i.d.\ code sequence (distribution $p$).
The arrival counts $N_i(t)$ belonging to code $i$ are independent and satisfy
\[
N_i(t)\sim\mathrm{Poisson}(\lambda p_i t)
\]
(Poisson splitting; \cite{DaleyVereJones2003,Kingman1993}).
Hence
\[
\Prb(i\in S(t))=\Prb(N_i(t)\ge 1)=:q_i(t)=1-e^{-\lambda p_i t}.
\]
Therefore
\begin{equation}
\E[\tau(t)]=\sum_{i=1}^K q_i(t)
=\sum_{i=1}^K \left(1-e^{-\lambda p_i t}\right),
\label{eq:Etau_general}
\end{equation}
and differentiating,
\begin{equation}
\frac{d}{dt}\E[\tau(t)]
=\lambda\sum_{i=1}^K p_i e^{-\lambda p_i t}.
\label{eq:speed_general}
\end{equation}

\begin{theorem}[Slowdown and saturation (general $p$)]
Under Poisson arrivals, the mean advancing speed $d\E[\tau(t)]/dt$ decreases monotonically in $t$, and
\[
\lim_{t\to\infty}\frac{d}{dt}\E[\tau(t)]=0
\]
holds. Moreover $\E[\tau(t)]\le K$ and
\[
\lim_{t\to\infty}\E[\tau(t)]=K
\]
holds. Consequently,
\[
\lim_{t\to\infty}\frac{\E[\tau(t)]}{t}=0.
\]
\end{theorem}

\begin{proof}
Each term $p_i e^{-\lambda p_i t}$ in \eqref{eq:speed_general} decreases monotonically as $t$ increases, and therefore the sum also decreases monotonically. Since the limit of each term is $0$, the limit of the sum is also $0$.
Also, from $0\le q_i(t)\le 1$, \eqref{eq:Etau_general} implies $\E[\tau(t)]\le K$.
Furthermore, as $t\to\infty$, $q_i(t)\to 1$, hence $\E[\tau(t)]\to K$.
Finally, since $\E[\tau(t)]$ is bounded, $\E[\tau(t)]/t\to 0$.
\end{proof}

Each term of \eqref{eq:speed_general} corresponds to the ``fraction of codes that remain unobserved.'' As time proceeds, the list of observed codes grows, so unobserved codes decrease, and the advancing speed of internal time decreases. This slowdown appears in an explicitly writable exponential form because the arrival process is Poisson. Since internal time $\tau(t)$ is the ``number of types of code observed, it cannot exceed $K$. When time is sufficiently large, every code is observed at least once, so the mean approaches $K$ and saturates. Therefore, on long time scales, no matter how much physical time advances, the internal time hardly advances, and the internal time becomes relatively shorter than the physical time.

\subsection{Uniform distribution {$p_i=1/K$}: closed form}
In the uniform case $p_i=1/K$, all $q_i(t)$ are equal:
\[
q(t)=1-e^{-\lambda t/K}.
\]
Therefore,
\begin{equation}
\E[\tau(t)]=K\bigl(1-e^{-\lambda t/K}\bigr),
\label{eq:Etau_uniform}
\end{equation}
\begin{equation}
\frac{d}{dt}\E[\tau(t)]=\lambda e^{-\lambda t/K},
\label{eq:speed_uniform}
\end{equation}
and the mean advancing speed decreases exponentially.

\paragraph{Time scale}
Equation \eqref{eq:Etau_uniform} shows that the natural time scale is $K/\lambda$.
The larger $K$ is, the longer physical time is needed until saturation, and the larger $\lambda$ is, the earlier saturation occurs.
This $K/\lambda$ becomes the reference when discussing the correspondence between physical time and internal time in a later section.

\section{Weighted internal time: a coding-consistent generalization}
\label{sec:weighted}

Internal time $\tau(t)$ is the minimal definition that counts the ``number of kinds of observed codes.''
On the other hand, even for the same ``observation of a new code,'' observation of a rare code carries a larger amount of information.
To incorporate this point, we assign weights (description lengths) to codes and introduce, as another internal time, the total weight contained in the list of observed codes.
Assign a positive weight (description length) $\ell_i>0$ to each code $i$, and define
\begin{equation}
T(t):=\sum_{i\in S(t)} \ell_i
\label{eq:T_def}
\end{equation}
as \textbf{weighted internal time}.
$T(t)$ corresponds to ``the total description length of codes contained in the list of observed codes,'' and while $\tau(t)$ counts only the ``number of kinds,'' $T(t)$ advances in a way that depends on ``which codes were observed.''

\subsection{Mean and slowdown under Poisson arrivals}
Under Poisson arrivals, since $q_i(t)=\Prb(i\in S(t))=1-e^{-\lambda p_i t}$,
\begin{equation}
\E[T(t)] = \sum_{i=1}^K \ell_i\, q_i(t)
= \sum_{i=1}^K \ell_i\left(1-e^{-\lambda p_i t}\right),
\label{eq:ET_general}
\end{equation}
and the mean advancing speed is
\begin{equation}
\frac{d}{dt}\E[T(t)] = \lambda\sum_{i=1}^K \ell_i p_i e^{-\lambda p_i t}.
\label{eq:ET_speed}
\end{equation}
Thus each term decreases monotonically as $t$ increases, and so does the sum.
Hence, weighted internal time also slows down in mean. For finite approaches $K$, $\E[T(t)]$, such as $t\to\infty$,
\[
\lim_{t\to\infty}\E[T(t)]=\sum_{i=1}^K \ell_i.
\]
That is, for finitely many codes, $T(t)$ also saturates toward an upper limit. In this sense, the saturation of $\tau(t)$ is saturation due to ``the upper limit of the number of kinds,'' while the saturation of $T(t)$ is saturation due to ``the upper limit of total description length.''

\subsection{Coding-consistency condition (requirement of entropy coding)}
As one way to choose weights $\ell_i$ consistently with the code distribution $p_i$, assume
\begin{equation}
-\log p_i = \beta \ell_i + \mathrm{const}, \qquad \beta>0,
\label{eq:coding_relation}
\end{equation}
This is a requirement that ``rarer codes have larger description length,'' and is consistent with coding in information theory (the correspondence between self-information $-\log p_i$ and description length) \cite{Shannon1948,CoverThomas2006}.

Ignoring the constant term \eqref{eq:coding_relation} means that $\ell_i$ is proportional to $-\log p_i$.
In this case, $T(t)$ can be interpreted as a natural representative quantity for ``the total self-information retained by the list of observed codes.''

\paragraph{Consistency with the uniform distribution}
For the uniform distribution $p_i=1/K$, $-\log p_i=\log K$ is constant independently of $i$.
Therefore, \eqref{eq:coding_relation} requires that all $\ell_i$ be equal.
Then $T(t)$ becomes a constant multiple of $\tau(t)$, and the weighted internal time is a generalization that contains ``internal time as the number of kinds.''

\paragraph{Implication for non-uniform distributions}
For non-uniform distributions, rarer codes have larger $\ell_i$, so if a rare code is observed in a later stage, $T(t)$ increases largely.
 Rare codes are harder to observe, so their contributions appear late. Therefore, $T(t)$ tends to have a form in which ``it advances relatively fast in an early stage, but the later-stage growth remains for a long time.'' This property appears more intuitively in a later section on Zipf-law input.

\section{Correspondence between physical time and internal time}
\label{sec:subjective_time}

We first give a closed form in the uniform case, and then describe the distortion for non-uniform distributions.

\subsection{Uniform reference system: nonlinear correspondence between $t$ and $\tau$}
In the uniform case, the mean of internal time is given by \eqref{eq:Etau_uniform}.
In an early stage, $\E[\tau(t)]\approx \lambda t$ increases almost linearly, but as time proceeds, the increase slows down and $\E[\tau(t)]\to K$ saturates.

\begin{figure}[t]
  \centering
  \includegraphics[width=\columnwidth]{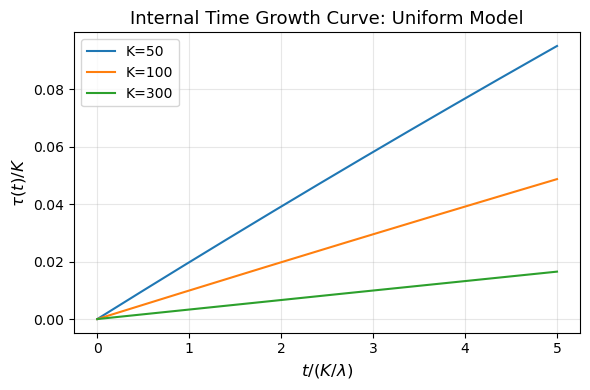}
  \caption{
    Correspondence between the mean internal time $\E[\tau(t)]$ and physical time $t$ in the uniform case.
    It increases rapidly in an early stage, and saturates in a later stage so that the advance of internal time is strongly slowed down.
  }
  \label{fig:tau_vs_t}
\end{figure}

\paragraph{What does ``relatively shorter'' mean?}
Since $\E[\tau(t)]$ is bounded, $\E[\tau(t)]/t\to 0$ holds as $t\to\infty$.
This means that ``on long time scales, even if physical time advances, internal time hardly advances.''
In this sense, internal time becomes relatively shorter than physical time.

\subsection{Inverse mapping: physical time required for an internal-time level}
Inverting \eqref{eq:Etau_uniform}, the physical time required to reach an internal-time level $\tau$ is
\begin{equation}
t(\tau)=\frac{K}{\lambda}
\log\!\left(\frac{K}{K-\tau}\right)
\label{eq:conversion}
\end{equation}
(which diverges as $\tau\to K$).
The physical-time increment corresponding to one unit of internal time is
\begin{equation}
\Delta t_{\mathrm{phys}}(\tau)=
t(\tau+1)-t(\tau)
=\frac{K}{\lambda}
\log\!\left(
\frac{K-\tau}{K-\tau-1}
\right),
\label{eq:dt_phys}
\end{equation}
which increases sharply as $\tau\to K$.

\begin{figure}[t]
  \centering
  \includegraphics[width=\columnwidth]{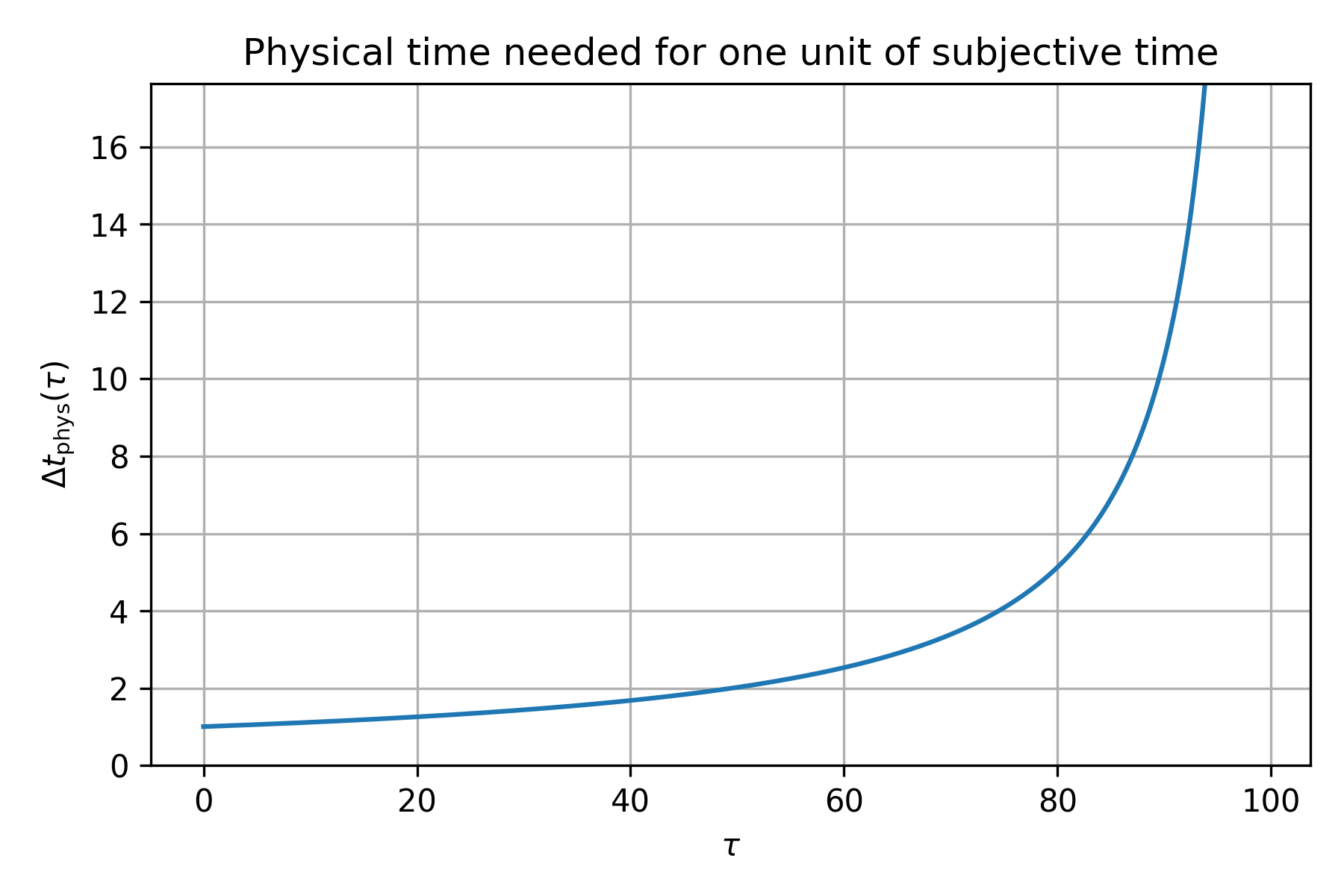}
  \caption{
    Physical time $\Delta t_{\mathrm{phys}}(\tau)$ required to obtain ``one unit of internal time'' in the uniform case (Eq.~\eqref{eq:dt_phys}).
    It increases sharply in later stages, showing that the advance of internal time is strongly slowed down.
  }
  \label{fig:dt_phys}
\end{figure}

\paragraph{Rephrasing}
Equation \eqref{eq:dt_phys} concretely expresses that the ``physical time needed to advance internal time by the same amount'' increases in later stages. This ``increase of required physical time'' is the most direct formalization of the claim that internal time becomes relatively shorter than physical time.

\subsection{Non-uniform distributions: distortion of $\E[\tau(t)]$}
For a general distribution $p$, \eqref{eq:Etau_general} holds.
Then the initial slope is
\[
\left.\frac{d}{dt}\E[\tau(t)]\right|_{t=0}=\lambda\sum_{i=1}^K p_i=\lambda,
\]
so in an early stage it increases approximately as $\lambda t$ independently of the distribution. On the other hand, the later-stage shape depends strongly on $p$.

Since high-probability codes are observed early and low-probability codes remain unobserved until late times, compared with the uniform case, a distortion can occur in which ``the early rise is fast, but the later tail is long.'' This distortion arises because, depending on the magnitude of $\lambda p_i t$, each code has a different rise in the observation probability $q_i(t)=1-e^{-\lambda p_i t}$.

\subsection{Heavy-tailed (Zipf-law) input: delayed saturation and a long tail}
If the statistics of information input has a heavy tail, the way internal time saturates changes significantly.
As a representative example of a heavy tail, consider the Zipf law (a power-law rank-frequency distribution) \cite{Newman2005,Zipf1949}.
Here, for rank $i$ of a code ($i=1$ is the most frequent; larger $i$ is rarer), assume the truncated Zipf distribution
\begin{equation}
p_i=\frac{i^{-\alpha}}{\sum_{j=1}^K j^{-\alpha}},\qquad i=1,\dots,K,\ \alpha>0
\label{eq:zipf}
\end{equation}
(for practical notes on estimation for truncated distributions, see \cite{ClausetShaliziNewman2009}).

Under Poisson arrivals, the probability that code $i$ has been observed at least once up to time $t$ is
$q_i(t)=1-e^{-\lambda p_i t}$, and therefore the mean internal time is
\begin{equation}
\E[\tau(t)] = \sum_{i=1}^K \left(1-e^{-\lambda p_i t}\right)
\label{eq:Etau_zipf}
\end{equation}
(a concrete example of \eqref{eq:Etau_general}). Under the Zipf law, there are many rare codes with small $p_i$, so they tend to remain unobserved over long periods. As a result, while $\E[\tau(t)]$ increases relatively fast in an early stage, in a later stage, the ``last few kinds'' are hard to fill, and the tail toward saturation becomes long.

\paragraph{Increase in an intermediate time window (qualitative scaling)}
To intuitively read the increase in an intermediate time window from \eqref{eq:Etau_zipf}, split contributions by whether $\lambda p_i t$ exceeds $1$.
For codes up to rank $i$ satisfying $\lambda p_i t \gtrsim 1$, we have $q_i(t)\approx 1$, while for rarer codes we have $q_i(t)\approx \lambda p_i t$.
Under Zipf law $p_i\propto i^{-\alpha}$, the boundary rank $i_c(t)$ tends to satisfy approximately
\begin{equation}
  \lambda p_{i_c} t \approx 1  
\end{equation}
then,

$i_c(t)$\ \text{tends to increase in proportion to a power of}\ $t$.

Therefore, in an intermediate time window well before saturation, $\E[\tau(t)]$ can show a slower-than-linear increase (a power-like increase) in $t$.
In particular, as $\alpha$ increases, probability concentrates on the top few codes, so the early rise becomes fast, but observation of rare codes becomes even more difficult, and the later tail tends to become longer.

\begin{figure}[t]
  \centering
  \includegraphics[width=\columnwidth]{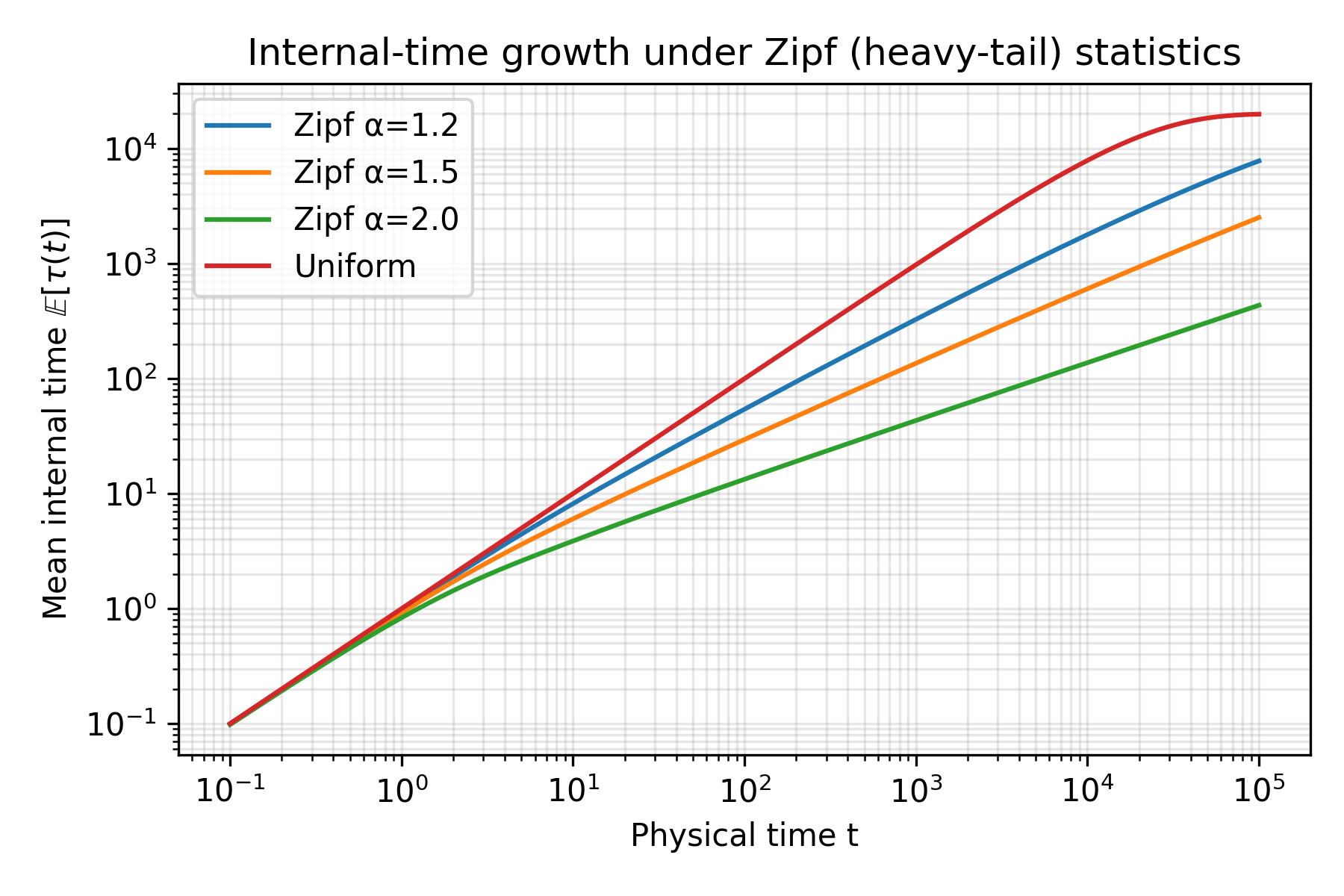}
  \caption{
    Mean internal time $\E[\tau(t)]$ in Zipf-law input (Eq.~\eqref{eq:zipf}) (Eq.~\eqref{eq:Etau_zipf}).
    Due to the heavy tail, the later-stage tail toward saturation tends to become long.
  }
  \label{fig:tau_zipf}
\end{figure}


\label{sec:conditional_entropy}

Internal time $\tau(t)$ retains only the ``number of kinds of observed codes,''
so even for the same $\tau(t)$, ``which codes have been observed'' is not determined in general.
As a quantity to measure what remains, we use the conditional entropy $H(S(t)\mid \tau(t))$. This conditional entropy is introduced to visualize, in an auxiliary way, ``what is lost when only internal time is retained.''

\subsection{Definition and identity}
Let the code sequence up to time $t$ be $X^{N(t)}=(X_1,\dots,X_{N(t)})$.
To compare different reductions, define the following two conditional entropies:
\[
L_A(t):=H\!\left(X^{N(t)}\mid S(t)\right),\qquad
L_B(t):=H\!\left(X^{N(t)}\mid \tau(t)\right),
\]
where $H(\cdot)$ is the Shannon entropy (natural logarithm) \cite{Shannon1948,CoverThomas2006}.
Introduce the difference
\[
U(t):=L_B(t)-L_A(t).
\]
Then, by the chain rule, the following holds.

\begin{theorem}[Identity]
Under any probabilistic model,
\begin{equation}
U(t)=H\!\left(S(t)\mid\tau(t)\right).
\label{eq:identity}
\end{equation}
\end{theorem}

\begin{proof}
$\tau(t)$ is a function of $S(t)$, and $S(t)$ is a function of $X^{N(t)}$.
By the chain rule,
\[
H\!\left(X^{N(t)},S(t)\mid\tau(t)\right)
=H\!\left(S(t)\mid\tau(t)\right)+H\!\left(X^{N(t)}\mid S(t),\tau(t)\right).
\]
On the other hand, since $S(t)$ is a function of $X^{N(t)}$, from $H(S(t)\mid X^{N(t)},\tau(t))=0$ we have
\[
H\!\left(X^{N(t)}\mid\tau(t)\right)
=H\!\left(X^{N(t)},S(t)\mid\tau(t)\right).
\]
Also, since $\tau(t)$ is a function of $S(t)$,
$H(X^{N(t)}\mid S(t),\tau(t))=H(X^{N(t)}\mid S(t))$.
Therefore,
\[
H(X^{N(t)}\mid\tau(t))=H(S(t)\mid\tau(t))+H(X^{N(t)}\mid S(t)),
\]
and subtracting $H(X^{N(t)}\mid S(t))$ from both sides yields the conclusion.
\end{proof}

\paragraph{Meaning of the identity}
$U(t)$ measures the difference between ``keeping the list of observed codes $S(t)$'' and ``keeping only internal time $\tau(t)$.'' The identity \eqref{eq:identity} shows that this difference coincides exactly with ``how much uncertainty remains in the list of observed codes when only internal time is known.'' Therefore, $H(S(t)\mid\tau(t))$ is a natural quantity that describes the information dropped by the reduction to internal time.

\subsection{Uniform distribution: closed form}
Consider the case where arrivals are Poisson and codes are uniform $p_i=1/K$.
By Poisson splitting,
\begin{equation}
N_i(t)\sim \mathrm{Poisson}\!\left(\frac{\lambda t}{K}\right)\ \text{independently}
\end{equation}
\begin{equation}
\Prb(i\in S(t))=:q(t)=1-e^{-\lambda t/K}.
\end{equation}
Therefore,
\[
\tau(t)=\sum_{i=1}^K \mathbf{1}\{N_i(t)\ge 1\}\sim \mathrm{Binomial}\!\bigl(K,q(t)\bigr).
\]

\begin{theorem}[Conditional equiprobability under the uniform distribution]
Under the uniform distribution, conditioning on $\tau(t)=m$, all ``candidates obtained by choosing $m$ codes from $K$ codes'' are equally likely.
Therefore,
\begin{equation}
H\!\left(S(t)\mid\tau(t)=m\right)=\log\binom{K}{m},
\label{eq:cond_uniform}
\end{equation}
and
\begin{equation}
H\!\left(S(t)\mid\tau(t)\right)=\E\!\left[\log\binom{K}{\tau(t)}\right]
\label{eq:U_uniform}
\end{equation}
holds.
\end{theorem}

\begin{lemma}[Symmetry]
Writing \eqref{eq:U_uniform} as a function of $q\in[0,1]$ by $U(q)$, we have
\[
U(q)=U(1-q).
\]
\end{lemma}

\begin{theorem}[Unimodality and maximization time]
Under the uniform distribution, $H(S(t)\mid\tau(t))$ is unimodal in $t$, and the unique maximizer is
\begin{equation}
t^*=\frac{K}{\lambda}\log 2
\label{eq:tstar}
\end{equation}
\end{theorem}

$t^*$ is the time at which the ancillary quantity $H(S\mid\tau)$ is maximized, and indicates the intermediate stage in which ``when only internal time is known, the candidates for the list of observed codes are the most numerous.'' The main result (slowdown and saturation) is a long-time-scale structure that includes the behavior as $t\to\infty$, and $t^*$ quantifies a different viewpoint (identification uncertainty of the reduction).

\subsection{Uniform distribution: asymptotic evaluation (Stirling + concentration)}
In the uniform case, $\tau\sim\mathrm{Bin}(K,q)$, $q=q(t)=1-e^{-\lambda t/K}$.
Let $m=Kx$. By Stirling approximation for the binomial coefficient (as a standard probability text, \cite{Feller1957}),
\begin{equation}
    \log\binom{K}{Kx}
=K\,\hb(x)-\frac12\log\!\bigl(2\pi K x(1-x)\bigr)+o(1),
\end{equation}
\begin{equation}
 \hb(x):=-x\log x-(1-x)\log(1-x).   
\end{equation}
Furthermore, since $\tau$ concentrates around mean $\mu=Kq$ and variance $\sigma^2=Kq(1-q)$, applying a second-order expansion (delta method) to the concave function $f(m)=\log\binom{K}{m}$ yields
\begin{equation}
H(S(t)\mid\tau(t))=\E[f(\tau)]
= f(\mu)+\frac12 f''(\mu)\sigma^2+o(1).
\end{equation}
Using $f''(\mu)\approx-\{(Kq)^{-1}+(K(1-q))^{-1}\}=-1/(Kq(1-q))$,
\begin{equation}
H(S(t)\mid\tau(t))
=K\,\hb(q)
-\frac12\log\!\bigl(2\pi K q(1-q)\bigr)
-\frac12
+o(1),
\end{equation}

 $(K\to\infty,\ q\in(\varepsilon,1-\varepsilon))$. 

In the later stage $t\to\infty$ for finite $K$, $q\to 1$ and $\tau\to K$ concentrate, 
and since $\log\binom{K}{\tau}\to 0$, we obtain $H(S(t)\mid\tau(t))\to 0$.

\section{Non-uniform distributions: upper bounds, equality conditions, gap representation, and approximation}
\label{sec:nonuniform}

In this section we organize how $H(S(t)\mid\tau(t))$ changes when the code distribution is non-uniform.
Here it is again important that this section concerns properties of an ancillary quantity, not the main result (slowdown and saturation).
The main result itself has already been shown in \S\ref{sec:slowdown}, and non-uniformity appears there as a distortion of $\E[\tau(t)]$ (see \S\ref{sec:subjective_time}).

\subsection{Universal upper bound (model independent)}
For any $m$, under $\tau(t)=m$, the number of candidates that $S(t)$ can take is $\binom{K}{m}$.
Therefore,
\begin{equation}
H(S(t)\mid\tau(t)=m)\le \log\binom{K}{m},
\end{equation}
and
\begin{equation}
H(S(t)\mid\tau(t))\le \E\!\left[\log\binom{K}{\tau(t)}\right]
\label{eq:upper_universal}
\end{equation}
holds. Equality is possible only when ``candidates are equally likely under $\tau(t)=m$.''

\subsection{Equality condition under Poisson arrivals}
Assume Poisson arrivals and an i.i.d.\ code sequence.
Then
\begin{equation}
N_i(t)\sim\mathrm{Poisson}(\lambda p_i t)\ \text{independently},
\end{equation}
\begin{equation}
q_i(t):=\Prb(i\in S(t))=1-e^{-\lambda p_i t}.
\end{equation}
Thus $I_i(t)=\mathbf{1}\{i\in S(t)\}$ are independent Bernoulli$(q_i)$.
For any $A\subseteq\{1,\dots,K\}$,
\[
\Prb(S(t)=A)=\prod_{i\in A} q_i\prod_{j\notin A}(1-q_j),
\]
so conditioning on $\tau(t)=m$,
\[
\Prb(S(t)=A\mid\tau(t)=m)\ \propto\ \prod_{i\in A}\frac{q_i}{1-q_i}.
\]

\begin{theorem}[Equality condition: the uniform distribution is the unique achiever (Poisson arrivals)]
Under Poisson arrivals, for any $m\in\{1,\dots,K-1\}$,
\[
H(S(t)\mid\tau(t)=m)=\log\binom{K}{m}
\]
holds if and only if $q_1(t)=\cdots=q_K(t)$.
Furthermore, by monotonicity of $q_i(t)=1-e^{-\lambda p_i t}$, this is equivalent to
\[
p_1=\cdots=p_K=1/K.
\]
\end{theorem}

The upper bound \eqref{eq:upper_universal} is achieved when ``candidates are spread as equally likely as possible.''
Under Poisson arrivals, such a situation occurs only when all $q_i(t)$ are equal, which is equivalent to uniformity of $p_i$.
Hence for non-uniform distributions, candidates are biased when only internal time is known, and the conditional entropy is smaller than in the uniform case.

\subsection{Gap from the upper bound: conditional KL representation}
Define the difference from the upper bound by
\[
G(t):=\E\!\left[\log\binom{K}{\tau(t)}\right]-H(S(t)\mid\tau(t)).
\]
Let $U_m$ be the equiprobable distribution over ``candidates obtained by choosing $m$ from $K$,''
and let $P_m(A):=\Prb(S(t)=A\mid\tau(t)=m)$. Then
\[
\log\binom{K}{m}-H(S\mid\tau=m)=\KL(P_m\|U_m),
\]
hence
\begin{equation}
G(t)=\E\!\left[\KL(P_{\tau(t)}\|U_{\tau(t)})\right]\ge 0.
\label{eq:gapKL}
\end{equation}

The gap $G(t)$ measures, on average, ``how far candidates deviate from being equally likely.''
For non-uniform distributions, $G(t)>0$ because easily observed codes tend to enter the list of observed codes, producing bias toward combinations that are easy to enter.

\subsection{Decomposition identity and approximation (Poisson-binomial)}
Under Poisson arrivals, since $I_i\sim\mathrm{Bernoulli}(q_i(t))$ are independent,
\[
H(S(t))=H(I_1,\dots,I_K)=\sum_{i=1}^K \hb(q_i(t)).
\]
From the chain rule $H(S)=H(\tau)+H(S\mid\tau)$ and \eqref{eq:identity}, we have:

\begin{theorem}[Decomposition identity]
Under Poisson arrivals,
\begin{equation}
H(S(t)\mid\tau(t))=\sum_{i=1}^K \hb(q_i(t)) - H(\tau(t)),
\qquad q_i(t)=1-e^{-\lambda p_i t}.
\label{eq:decomp}
\end{equation}
\end{theorem}

Here $\tau(t)=\sum I_i$ has a Poisson-binomial distribution (computations and approximations: \cite{Ehm1991,Hong2013}).

\begin{equation}
    \mu(t)=\E[\tau(t)]=\sum q_i(t),
\end{equation}
\qquad\begin{equation}
\sigma^2(t)=\Var(\tau(t))=\sum q_i(t)(1-q_i(t)).
\end{equation}
In a region where $\sigma^2(t)$ is sufficiently large,
\[
H(\tau(t))\approx \frac12\log(2\pi e\,\sigma^2(t))
\]
can be used as an approximation \cite{CoverThomas2006}. Therefore,
\begin{equation}
H(S(t)\mid\tau(t))\approx \sum_{i=1}^K \hb(q_i(t))-\frac12\log(2\pi e\,\sigma^2(t)).
\label{eq:approx_normal}
\end{equation}

In an early stage where $\lambda t$ is small, since $q_i(t)\approx \lambda p_i t$,
\[
\sum_{i=1}^K \hb(q_i(t))
\approx \lambda t\Bigl[\log\frac{1}{\lambda t}+1+H(p)\Bigr],
\qquad
\sigma^2(t)\approx \lambda t
\]
is obtained, and $H(p)$ appears in the initial growth rate.

\subsection{Maximization time}
\paragraph{Uniform distribution (exact)}
In the uniform case, $H(S(t)\mid\tau(t))$ is unimodal, and the maximization time is given by \eqref{eq:tstar}.

\paragraph{Non-uniform distributions (no closed form in general)}
For non-uniform distributions, the maximization time generally has no closed form.
One maximizes \eqref{eq:decomp} directly, or identifies it approximately by the first-order condition of the approximation \eqref{eq:approx_normal}.

\section{Discussion}
\label{sec:discussion}

In this section we organize, from both intuition and equations, from which assumptions the main result ``slowdown and saturation of internal time'' arises and by which features of input statistics it is controlled. We also describe the meaning of the long tail produced by heavy-tailed (Zipf-law) input, and the positioning of the generalization by weighted internal time. For the conditional entropy introduced as an ancillary quantity, we clarify that it is not the core of this paper, and describe in what range it is useful as an auxiliary index.

\subsection{Reconfirming the core of this paper: slowdown and saturation are the conclusion}
The central conclusion of this paper is that internal time $\tau(t)$ (the number of kinds of observed codes) ticked by a response system is regulated by the input statistics (arrival rate $\lambda$ and code distribution $p$), that the mean advancing speed decreases monotonically with time, and that it saturates under finite $K$. This conclusion follows directly in the standard setting of Poisson arrivals and an independent and identically distributed code sequence: the mean internal time is given by \eqref{eq:Etau_general}, and its time derivative (mean advancing speed) has the form \eqref{eq:speed_general}. That is, the advancing speed is a sum of $p_i e^{-\lambda p_i t}$, and since each term decreases monotonically with $t$, the mean advancing speed necessarily slows down; moreover, since $q_i(t)\to 1$, $\E[\tau(t)]\to K$ and saturates.

In this sense, the statement that ``internal time becomes relatively shorter than physical time'' is not a metaphor, but can be expressed as the rigorous consequence
$\lim_{t\to\infty}\E[\tau(t)]/t=0$ that follows from boundedness of $\E[\tau(t)]$.
This conclusion holds for a general distribution $p$, and each distribution changes ``how it slows down and how the tail remains.''

\subsection{Input statistics determine the slowdown curve: $\lambda$ is the time scale and $p$ is the shape}
From \eqref{eq:speed_general}, the arrival rate $\lambda$ mainly determines the ``time scale.''
For example, if one nondimensionalizes time by $u=\lambda t$, the basic form of slowdown is described by $u$.
On the other hand, the code distribution $p$ determines the \emph{shape} of slowdown.
Since high-probability codes are likely to be observed early and low-probability codes tend to remain unobserved for a long time, the allocation of ``early rise,'' ``middle-stage growth,'' and ``late-stage tail'' of $\E[\tau(t)]$ changes with $p$.
This is an advantage of the present model in the sense that, by restricting internal time to the minimal record of ``number of kinds,'' input statistics become directly visible.

In the uniform case, moreover, the correspondence between physical time and internal time can be written in closed form.
The inverse mapping \eqref{eq:conversion} shows that the physical time required to achieve internal-time level $\tau$ diverges as $\tau\to K$, and the increment \eqref{eq:dt_phys} gives concretely that ``the physical time to advance internal time by one unit'' increases in later stages.
This expression that ``required physical time increases in later stages'' conveys the content of slowdown and saturation most intuitively.

\subsection{Heavy-tailed (Zipf-law) input: a long tail is not an ``exception'' but a consequence of statistics}
For heavy-tailed distributions, since there are many rare codes, the tail toward saturation becomes long in later stages.
In this paper we considered the Zipf law as a representative example (overview \cite{Newman2005}, original \cite{Zipf1949}).
Even under the truncated Zipf distribution \eqref{eq:zipf}, the mean internal time is given as a concrete form of the general formula \eqref{eq:Etau_general} by \eqref{eq:Etau_zipf}, and the core itself (slowdown and saturation) does not change.
What changes is the \emph{time allocation} toward saturation.

In particular, under the Zipf law, since there are many codes with small probabilities,
the phenomenon that ``unobserved codes remain even after quite a long time'' naturally arises.
This implies that, even for finite $K$, saturation can be extremely delayed, and can serve as a minimal model to explain the impression in real data (natural language, behavioral sequences, etc.) that ``new things never fully run out.''
Since estimation and truncation handling for heavy-tailed distributions require statistical care, the discussion in \cite{ClausetShaliziNewman2009} is useful when applying to real data.

\subsection{Weighted internal time: a natural generalization of the number-of-kinds model}
In this paper we adopted $\tau(t)$ (number of kinds) as the minimal definition of internal time, but to incorporate the standard viewpoint that ``rarer codes carry larger information,'' the weighted internal time $T(t)$ is natural (\S\ref{sec:weighted}).
In particular, assuming \eqref{eq:coding_relation} between weights $\ell_i$ and distribution $p_i$, $\ell_i$ is consistent with self-information $-\log p_i$, and $T(t)$ can be interpreted as a representative quantity of ``the amount of information retained by the list of observed codes.''
This consistency is based on basic principles of information theory \cite{Shannon1948,CoverThomas2006}, and provides a consistent generalization that includes the number-of-kinds model ($\tau$).

The mean of weighted internal time is given by \eqref{eq:ET_general} and its advancing speed by \eqref{eq:ET_speed}, and similarly to $\tau(t)$, it slows down in mean and saturates for finite $K$.
However, since the contribution is larger when a rare code is observed in a later stage, for heavy-tailed input, the ``later-stage tail'' of $T(t)$ may be emphasized more strongly than that of $\tau(t)$.
This point is important in situations where one wants to define internal time not only by number of kinds but also as an amount of information.

\subsection{Positioning of the ancillary quantity (conditional entropy): what the reduction to internal time drops}
The conditional entropy $H(S(t)\mid\tau(t))$ is not an assumption nor a conclusion that supports the core of this paper (slowdown and saturation), but an auxiliary quantity that measures ``how much uncertainty remains in the list of observed codes when only internal time is known.''
The identity \eqref{eq:identity} shows that the reduction difference $U(t)$ equals this conditional entropy, organizing the information dropped by the reduction to internal time from the basic structure of the chain rule of entropy \cite{Shannon1948,CoverThomas2006}.

In the uniform case, candidates become equiprobable, yielding the closed form \eqref{eq:U_uniform}, and unimodality and the maximization time \eqref{eq:tstar} follow.
In the non-uniform case, the universal upper bound \eqref{eq:upper_universal} and equality conditions can be organized, and the difference from the upper bound is expressed as the mean KL divergence \eqref{eq:gapKL}.
This section is not intended to replace ``slowdown and saturation,'' but is positioned as an auxiliary result that quantifies what becomes indistinguishable when adopting internal time as a minimal record.

\subsection{Limitations and extensions: scope of the minimal model}
Since this paper intentionally adheres to a minimal model, it has the following limitations.
First, internal time is defined by ``the number of kinds of observed codes (or a weighted sum),'' and it does not retain more detailed history such as frequencies, order, or inter-arrival intervals.
Second, the arrival process is restricted to Poisson for analysis.
Third, correspondence to external psychological indices or behavioral indices is not assumed.

On the other hand, directions of extension are clear.
(i) Extend the arrival process to non-Poisson (periodicity, self-exciting, etc.) and investigate the universality of slowdown and how it breaks (general theory of point processes: \cite{DaleyVereJones2003,Kingman1993}).
(ii) Slightly extend the internal state to include frequency information or recent history, design an internal time that includes such information, and investigate how it connects to the slowdown and saturation in this paper.
(iii) Consider heavy-tail limits and $K\to\infty$ and organize scaling in regions where saturation is effectively not observed.
These are natural tasks to examine to what extent the core of slowdown and saturation of internal time depending on input statistics holds generally.

\section{Conclusion}
This paper presented, for a response system that updates its internal state according to information input, a minimal model that defines as internal time the number of kinds of observed codes, and rigorously derived slowdown and saturation of internal time according to the statistics of information input.
Under Poisson arrivals, the mean advancing speed of internal time decreases monotonically, and if the number of kinds of codes is finite it saturates; therefore, on long time scales internal time becomes relatively shorter than physical time.
In the uniform case, we provided the inverse relation between physical time and internal time and explicitly showed that internal time becomes harder to advance in later stages.
As a generalization, we also introduced code-dependent weights (description lengths) and provided a formulation of internal time according to information amount.

As an ancillary quantity, we introduced conditional entropy and quantified the uncertainty that remains when only internal time is known.
We organized the closed form, unimodality, and maximization time in the uniform case, and upper bounds, equality conditions, gap representations, and approximate evaluations in the non-uniform case.
These are independent of the main result of slowdown and saturation, but they provide an auxiliary view of aspects of the information dropped by the reduction to internal time.

\appendix
\section{Appendix A: unimodality and maximizer in the uniform distribution}
\label{app:unimodality}

In the uniform case, $\tau\sim \mathrm{Bin}(K,q)$, and
\begin{equation}
U(q)=H(S\mid\tau)=\E\!\left[\log\binom{K}{\tau}\right]
\end{equation}
\begin{equation}
=\sum_{m=0}^K \binom{K}{m} q^m(1-q)^{K-m}\log\binom{K}{m}.
\label{eq:bernstein}
\end{equation}
This is the Bernstein polynomial of the discrete concave function $f(m)=\log\binom{K}{m}$.

\paragraph{Symmetry}
From $\binom{K}{m}=\binom{K}{K-m}$ and symmetry of the binomial weights, \eqref{eq:bernstein} implies $U(q)=U(1-q)$.

\paragraph{Strict concavity (outline)}
The sequence $m\mapsto \log\binom{K}{m}$ is discretely strictly concave on $m=0,\dots,K$, and its second difference is negative for $1\le m\le K-1$.
As a standard property of the Bernstein operator, the second derivative of the Bernstein polynomial is represented as a positively weighted average of the second differences (see \cite{Lorentz1986}).
Therefore $U''(q)<0$ holds for $q\in(0,1)$, and $U(q)$ is strictly concave on $(0,1)$.
Hence it is unimodal.

\paragraph{Maximizer}
By strict concavity and symmetry with respect to $q=1/2$, the maximizer is uniquely $q^*=1/2$.
Since $q(t)=1-e^{-\lambda t/K}$,
\[
1-e^{-\lambda t^*/K} = \frac12,
\qquad
t^* = \frac{K}{\lambda}\log 2.
\]

\bibliographystyle{plainnat}
\bibliography{references}

\end{document}